	\documentclass[11pt, reqno]{amsart} 

	\newcommand{\Number}{section} 

	\numberwithin{equation}{section} 

	\usepackage{amssymb, amsfonts, verbatim} \usepackage{epsfig, amsmath, cancel} 
	\usepackage[mathscr]{eucal} \usepackage{graphicx} \usepackage{epsfig} 

		\newcommand{\Q}{\hspace{.065in}} 

			\newcommand{\del}[1]{\partial_{#1}}

		\newcommand{\ga}{\alpha}
		\newcommand{\gb}{\beta}
		
		\newcommand{\gd}{\delta}
		
		\newcommand{\gf}{\phi}
		
		\renewcommand{\gg}{\gamma}

		\newcommand{\gl}{\lambda}
		\newcommand{\gm}{\mu}
		
		\newcommand{\go}{\omega}
		
		\newcommand{\gq}{\theta}
		\newcommand{\gr}{\rho}
		
		\newcommand{\gs}{\sigma}
		\newcommand{\gt}{\tau}

		\newcommand{\gx}{\xi}
		\newcommand{\gy}{\psi}
		\newcommand{\gz}{\zeta}

			\newcommand{\gD}{\Delta}
			\newcommand{\gF}{\Phi}
			
			\newcommand{\gL}{\Lambda}

			\newcommand{\gS}{\Sigma}

			\newcommand{\gY}{\Psi}

			\newcommand{\scF}{\mathscr{F}}
			\newcommand{\scG}{\mathscr{G}}
			
			\newcommand{\scI}{\mathscr{I}}
			\newcommand{\scJ}{\mathscr{J}}
			\newcommand{\scK}{\mathscr{K}}

			\newcommand{\scY}{\mathscr{Y}}

			\newcommand{\bbN}{\mathbb{N}}

			\newcommand{\bbR}{\mathbb{R}}

		\theoremstyle{plain}
		
		\newtheorem{lemma}{Lemma}[\Number]
		\newtheorem{prop}{Proposition}[\Number]
		\newtheorem{theorem}{Theorem}[\Number]
		\newtheorem{corollary}[theorem]{Corollary}

			\theoremstyle{definition}

			\theoremstyle{remark}

	\newcommand{\mnote}[1]{}        
	
	\newcommand{\DNa}[2]{\underset{\hspace{-.5em}#1}{\nabla_{#2}}} 
	\newcommand{\Un}[1]{\underset{\hspace{0em}#1}{\nabla}} 
	\newcommand{\UL}[1]{\underset{\hspace{0em}#1}{\Delta}} 
	\newcommand{\Scal}[1]{\underset{\hspace{0em}#1}{R}} 
	\newcommand{\LW}[1]{\underset{\hspace{0.3em}#1}{LW}} 
	\newcommand{\LV}[1]{\underset{\hspace{0.3em}#1}{LV}} 
	
	\newcommand{\Ugf}[2]{\underset{\hspace{-0.1em}#1\phantom{#2}}{\gf^{#2}}} 
	\newcommand{\Dgf}[2]{\underset{\hspace{-0.75em}#1}{\gf_{#2}}} 
	\newcommand{\UDgf}[3]{\underset{\hspace{-0.75em}#1}{\gf^{#2}_{#3}}} 
	
	\newcommand{\Ugy}[2]{\underset{\hspace{-0.1em}#1\phantom{#2}}{\gy^{#2}}} 
		
	\newcommand{\UW}[1]{\underset{\hspace{-0.1em}#1}{W}} 
	\newcommand{\UV}[1]{\underset{\hspace{-0.1em}#1}{V}} 
	\newcommand{\UA}[2]{\underset{\hspace{-0.1em}#1\phantom{#2}}{A^{#2}}} 
	
	\newcommand{\Step}[2]{{\it \noindent {Step {#1}.  (#2)}}\\ \indent } 

\begin{document}

	\title[Near CMC]{Near-Constant Mean Curvature Solutions of the Einstein Constraint Equations with Non-Negative Yamabe Metrics} 
	
		 \author{Paul T Allen} 
	 \address{Albert Einstein Institute, Max Planck Institute for Gravitational Physics\\
	 D-14476 Golm\\
	Germany\\
	pallen@aei.mpg.de}
	
		\author{Adam Clausen}
	 \address{Department of Physics\\
	 Lawrence University\\
	 Appleton WI 54911, USA\\
	adam.clausen@lawrence.edu}

	\author{James Isenberg}
	\address{ Department of Mathematics\\ 
	University of Oregon\\
	Eugene, OR 97403, USA\\
	jim@newton.uoregon.edu}
	 
	 \date{\today}

	\begin{abstract} We show that sets of conformal data on closed manifolds with the metric in the positive or zero Yamabe class, and with the gradient of the mean curvature function sufficiently small, are mapped to solutions of the vacuum Einstein constraint equations.  This result extends previous work which required the conformal metric to be in the negative Yamabe class, and required the mean curvature function to be nonzero. \end{abstract}
	\maketitle


\section{Introduction}

The set of smooth, constant mean curvature (CMC) solutions of the vacuum Einstein constraint equations is fairly well understood.  For closed manifolds, there is a complete parameterization of these solutions in terms of conformal data \cite{I-CMC}.
For the asymptotically Euclidean and asymptotically hyperbolic cases, similar results hold \cite{CBIY-AE,AC-AH}.

Much less is know about non constant mean curvature solutions.  The mathematical reason for this is that while the CMC condition effectively eliminates three of the four Einstein constraint  equations from the analysis, in the non-CMC case one must handle the full, coupled system.

All of the non-CMC results to date \cite{IM-nonCMC,CBIY-AE,IP-AH} require that the gradient of the mean curvature $\gt$ be sufficiently small; we call such solutions ``near-CMC''.  In the case of closed manifolds, these results also require that the metric be in the negative Yamabe class, and that the mean curvature function have no zeroes.  While we have not yet managed to relax the small $|\nabla \gt|$ condition, in this paper we show that we can construct non-CMC solutions on closed manifolds with the metric in the positive or zero Yamabe class, and with the mean curvature function allowed to have zeroes in the positive Yamabe case.

The procedure we use for proving our results here is the semi-decoupled  sequence (constructive) method, which we have introduced in \cite{IM-nonCMC}.  The chief difference between our work here and \cite{IM-nonCMC} is that while we can use a sequence of constant sub and super solutions for sets of conformal data with negative Yamabe class metrics and $\gt$ nowhere zero, for solutions with positive or zero Yamabe class metrics we require non-constant sub solutions.  The focus in this paper is on how to obtain such sub solutions and how to control them.  We discuss this issue in Section \ref{S:Lich},
after a brief introduction to the conformal method in Section \ref{S:C-Method}.
Also in Section \ref{S:Lich}
we show that the Lichnerowicz  equation with negative Yamabe class metric and $\gt$ allowed to have zeroes (but not allowed to identically vanish) always admits solutions. \mnote{Do we really let $\gt$ have zeroes?}
In Section  4
we describe the semi-decoupling method for constructing near-CMC solutions of the constraints, and then state and prove our main theorems.  We make some concluding remarks in Section 5.
Note that, in this paper, we are not concerned with optimizing the regularity conditions on either the choice of conformal data or on the solutions of the constraints which we obtain. Presumably one could produce solutions with the same degree of roughness discussed in \cite{M-Rough} and \cite{CB-Constraints}.

\section{The Conformal Method and the Lichnerowicz Equation}\label{S:C-Method}
The Einstein vacuum constraint equations require that a set of initial data $(\gS; \gg, K)$ consisting of a Riemannian metric $\gg$ and a symmetric tensor $K$ specified on a three-dimensional manifold $\gS$, satisfy the equations
	\begin{equation}\label{FirstEin}
	{R} - K^{ab}K_{ab}+(K^a_{\Q a})^2 =0
	\end{equation}
and
	\begin{equation}\label{SecondEin}
	\DNa{}{a} K^{a}_{\Q b} - {\nabla_b} (K^a_{\Q a}) =0,
	\end{equation}
where the covariant derivative $\nabla$, the scalar curvature $R$, and all contractions and traces are calculated with respect to the metric $\gg$.

The idea of the conformal method is that one may construct and parameterize solutions of the constraints \eqref{FirstEin}-\eqref{SecondEin} by splitting $\gg$ and $K$ into a set of freely specified data and a set of determined data.  The freely specified ``conformal data'' consists of a Riemannian metric $\gl$, a symmetric tensor $\gs$ which is trace-free and divergence-free with respect to $\gl$, and a function $\gt$, all specified on a manifold $\gS$.  The determined data consists of a vector field $W$ and a positive definite scalar field $\gf$.  Using the conformal data to define the covariant derivative $\Un{\gl}$ together with the corresponding Laplacian $\UL{\gl}$, scalar curvature $\Scal{\gl}$, and conformal Killing operator 
	\begin{equation*}
	(\LW{\gl})_{ab} := \DNa{\gl}{a}W_b + \DNa{\gl}{b}W_a - \tfrac{2}{3} \gl_{ab} \DNa{\gl}{c} W^c,
	\end{equation*}
and to define contractions, we write out the constraint equations as follows:
	\begin{align}\label{Lich}
	\UL{\gl}\gf &= \tfrac18 \Scal{\gl}\gf - \tfrac18 | \gs + \LW{\gl}|^2 \gf^{-7} + \tfrac{1}{12} \gt^2 \gf^5 \\
	\Un{\gl}\cdot \LW{\gl} &= \frac23\gf^6 \, \Un{\gl}\gt. \label{LW}
	\end{align}

If, for a given set of conformal data $(\gS; \gl, \gs, \gt)$, equations \eqref{Lich}-\eqref{LW}
 can be solved for $(\gf, W)$, then one readily verifies that the reconstituted  data
	\begin{align}
	\label{gammarecon}
	\gg_{ab} &= \gf^4 \gl_{ab} \\
	K^{ab}&= \gf^{-10} (\gs + \LW{})^{ab} + \tfrac{2}{3} \gf^{-4} \gt\gl^{ab}
	\label{Krecon}
	\end{align}
satisfy the constraint equations.

Since $K^a_a = \gt$ and since $K_{ab}$ corresponds to the second fundamental form for the embedded Cauchy hypersurface $\gS$ in a spacetime development evolved from the initial data $(\gS; \gg, K)$, the function $\gt$ represents the mean curvature of the Cauchy surface.  Specifying $\gt = \text{constant}$ results in constant mean curvature (CMC) data.  This condition is important, since if we specify conformal data with constant mean curvature then the equations \eqref{Lich}-\eqref{LW} decouple.  Equation \eqref{LW} becomes a homogeneous linear elliptic equation for $W$ and we have $LW{}=0$ in all (compact) cases.  The determination of whether a particular set of CMC conformal data 
produces a solution of the constraints is thus determined entirely by the solubility of the non-linear elliptic ``Lichnerowicz'' equation \eqref{Lich}, with $LW{}=0$.

\section{Solving the Lichnerowicz Equation}\label{S:Lich}
In this section we discuss the solubility of the Lichnerowicz equation, independent of possible coupling to the other constraint equations.  (We return to the system \eqref{Lich}-\eqref{LW} in Section 4.)  To emphasize this, in \eqref{Lich} we replace the term $|\gs + LW|^2$  which involves the product of tensor fields, by the simple function $\gm^2$.  Thus we work with the Lichnerowicz equation in the form
	\begin{equation}\label{MuLich}
	\UL{\gl} \gf = \tfrac18 \Scal{\gl}\gf - \tfrac18 \gm^2 \gf^{-7} + \tfrac{1}{12}\gt^2 \gf^5.
	\end{equation}
Here $\gm$ and $\gt$ are arbitrary smooth functions, which may or may not have zeroes.

We may further simplify the analysis of the Lichnerowicz equation by  making use of its conformal covariance which tells us that there is a solution to \eqref{MuLich} for a given set of data $(\gS; \gl, \gm, \gt)$ if and only if there is a solution to \eqref{MuLich} for the related set of data $(\gS; \gq^4\gl, \gq^{-6}\gm, \gt)$ (see \cite{BI}).  Combining this property with the Yamabe theorem \cite{S-Scal} (see also \cite{LP-Yamabe}), we find that to determine the solubility of the Lichnerowicz equation for general sets of conformal data, it is sufficient to study \eqref{MuLich} for metrics having constant scalar curvature of either $+8$, $0$, or $-8$.\footnote{In fact, the full Yamabe Theorem is not needed; it suffices that each conformal class contain a metric with scalar curvature having definite sign; for a proof of this more elementary fact see \cite{A-NLbook}.}

The key tool we employ for proving the existence of solutions is the sub and super solution theorem.  The most useful version for our work here makes use of the Sobolev spaces $W^{2,p}$ and H\"older spaces $C^{k,\ga}$; see \cite{B} for definitions and properties of these function spaces.  For a proof of the theorem stated here, see \cite{IM-nonCMC}; the same result is proven for rougher data in \cite{M-Rough}.
	
\begin{theorem}\label{T:SubSuper}
Let $(\gS; \gl)$ be a closed Riemannian manifold with $C^2$ metric and let $f\in C^1(\gS\times\bbR_+)$.  Assume that there exists $\gf_-,\gf_+ :\gS\to \bbR_+$ such that with $p>3$ we have\footnote{We note that the inequalities stated in conditions (3) and (4) below involve Banach space elements, $\UL{\gl}\gf_-$ and $\UL{\gl}\gf_+$, and therefore are not strictly well-defined pointwise. These inequalities are presumed to hold  on any subset of the manifold $\gS$ with non-zero measure}
	\begin{enumerate}
	\item $\gf_{\pm} \in W^{2,p}(\gS)$,
	\item $0<\gf_-(x) \leq \gf_+(x)$ for all $x\in\gS$,
	\item $\UL{\gl}\gf_- \geq f(x,\gf_-)$, and
	\item $\UL{\gl}\gf_+ \leq f(x,\gf_+)$.
	\end{enumerate}
Then there exists $\gf:\gS\to\bbR_+$ such that
	\begin{enumerate}
	\item $\gf\in C^{2,\ga}(\gS)$ for $\ga\in(0,1-\tfrac{3}{p})$,
	\item $\gf_-(x) \leq \gf(x)\leq \gf_+(x)$ for all $x\in\gS$, and
	\item $\UL{\gl}\gf = f(x,\gf)$.
	\end{enumerate}
\end{theorem}
The functions $\gf_+$,$\gf_-$ are called super and sub solutions (resp.).  The bulk of the work required to obtain the results presented here lies in the construction of sub and super solutions for the Lichnerowicz equation \eqref{MuLich} by means of a technique which can be applied to the coupled system \eqref{Lich}-\eqref{LW}.  We first focus on \eqref{MuLich} for positive Yamabe metrics and show the following.

\begin{prop}\label{P:MuLichPlus}
Let $\gS$ be a closed manifold, let $\gl$ be a smooth, positive Yamabe class metric on $\gS$, and let $\gm$ and $\gt$ be smooth functions on $\gS$ with $\gm$ not identically zero.  Then there exists a unique smooth\footnote{The proof in \cite{I-CMC} produces a function $\gz_-\in W^{4,p}(\gS)$ for $p>3$; one readily bootstraps the argument to show that for smooth data, we obtain a smooth solution as well.}
 solution $\gf$ to the Lichnerowicz equation \eqref{MuLich}.
\end{prop}

\begin{proof}
As noted above, it is sufficient to prove existence and uniqueness of solutions for which $\Scal{\gl}$ is constant; thus we work with the equation
	\begin{equation}\label{LichPlus}
	\UL{\gl} \gf = \gf - \tfrac18 \gm^2 \gf^{-7} + \tfrac{1}{12} \gt^2 \gf^5. 
	\end{equation}

\Step{1}{Sub solution for $\gt=0$ case} In \cite{I-CMC} (see class $(\scY^+,\gs\neq 0, \gt=0)$ in Section 5)  it is shown that there exists a smooth function $\gz_-$ such that
	\begin{equation}\label{E:Sub1}
	\UL{\gl}\gz_- \geq \gz_- - \tfrac18 \gm^2 \gz_-^{-7}. 
	\end{equation}
For completeness, we summarize the argument presented there: Let
	\begin{equation*}
	A:= \max{\{1, \tfrac18\max_{\gS}{\gm^2}\}} 
	\end{equation*}
and consider the linear PDE
	\begin{equation}\label{SubLin1}
	\UL{\gl} \gz_- - \gz_- = -\tfrac18\gm^2 A^{-7}. 
	\end{equation}
If follows from the non-degeneracy of the operator $(\UL{\gl} -1)$ on compact $\gS$ and from the smoothness of $\gl$ and $\gm$ that there exists a unique, smooth solution $\gz_-$ to \eqref{SubLin1}.  To show that this function satisfies the inequality \eqref{E:Sub1} as well, and therefore is a sub solution for the Lichnerowicz equation with $\gt=0$, we first note that  $\tfrac18 \gm^2 A^{-7}$ is non-negative and is not identically zero.  Thus the maximum principle guarantees that $\gz_- >0$.  Next, since the function $G(x,s):= \tfrac18\gm^2 s^{-7}$ is monotonically non-increasing in $s$, and since by definition $A \geq 1$ and $A \geq \tfrac18\gm^2$, we have $G(x,A) \leq G(x,1)$ and  $G(x,A) \leq A$ for all $x\in\gS$.  The latter inequality, together with the definition of $\gz_-$, guarantees that $\gz_-$ satisfies the inequality
	\begin{equation}
	\UL{\gl}\gz_- - \gz_- \geq -A, 
	\end{equation}
from which we infer (via the maximum principle)\footnote{The version we use here appears as \#3 in \cite{I-CMC}.} 
that $\gz_- \leq A$.
Using this last inequality together with the monotonicity (in $s$) of $G$ to infer that $G(x,A) \leq G(x,\gz_-)$, and writing \eqref{SubLin1} as 
	\begin{equation}
	 \UL{\gl} \gz_- = \gz_- -G(x,A),
	\end{equation}
we verify  that indeed $\gz_-$ satisfies the sub solution inequality \eqref{E:Sub1}.

\Step{2}{General sub solution}
The inequality \eqref{E:Sub1} is not strict; in order to obtain a sub solution for the Lichnerowicz equation \eqref{LichPlus} with positive Yamabe metric and $\gt^2$ not identically zero, it is  useful to first replace $\gz_-$ by a function $\gx_-$ for which \eqref{E:Sub1} is a strict  inequality.  This is easily done by setting
	\begin{equation}
	\gx_- := \gz_- - \tfrac12 \min_{\gS}{\gz_-}. 
	\end{equation}
Here we use the continuity of $\gz_-$ and the fact  that $\gz_- >0$ to verify that $ \min_{\gS}{\gz_-}$ exists and is positive.  Consequently we have that
\mnote{Find missing constants here}
	\begin{equation}
	 \gz_-(x) > \gx_-(x) \geq \tfrac12 (\min_{\gS}{\gz_-}) >0, \quad x\in \gS.
	\end{equation}
Using the monotonicity of $G(x,\cdot)$ we have $G(x, \gx_-) \geq G(x, \gz_-)$ for all $x\in \gS$, from which it follows that
	\begin{equation}
	\begin{aligned}
	\UL{\gl}\gx_- - \gx_- + \frac18\gm^2 \gx_{-}^{-7} &\geq \UL{\gl} \gz_- +\tfrac12 (\min_{\gS}{\gz_-}) - \gz_- + \frac18\gm^2 \gz_{-}^{-7}\\
	&\geq \tfrac12 (\min_{\gS}{\gz_-})\\
	&>0.
	\end{aligned} 
	\end{equation}
Let us now multiply $\gx_-$ by a positive number $\gb\in(0,1)$ (to be determined later); we obtain
\mnote{$\gm$ or $\gb$?}
	\begin{equation}
	\UL{\gl}(\gb\gx_-) -  (\gb\gx_-) + \frac{\gb^8}{8} \gm^2 (\gb\gx_-)^{-7} \geq \tfrac12 \gb (\min_{\gS}{\gz_-}).
	\end{equation}
Since $\gb\in(0,1)$, one has $\tfrac18 \gm^2 \geq \tfrac18 \gb^8 \gm^2$; hence 
	\begin{equation}
	\UL{\gl}(\gb\gx_-) -  (\gb\gx_-) + \tfrac{1}{8} \gm^2 (\gb\gx_-)^{-7} \geq \tfrac12 \gb (\min_{\gS}{\gz_-}). 
	\end{equation}
We wish to make a choice of the constant $\gb$ so that if we set $\gf_- = \gb\gx_-$, then $\gf_-$ is a sub solution  for \eqref{LichPlus}.  This is accomplished by choosing $\gb\in(0,1)$ so that
	\begin{equation}
	\tfrac12 \gb (\min_{\gS}{\gz_-}) \geq \tfrac{1}{12} \gt^2 (\gb\gx_-)^5 .
	\end{equation}
One readily verifies that this last estimate is satisfied provided
	\begin{equation}
	\gb \leq \left[\frac{6(\min_{\gS}{\gz_-})}{(\max_{\gS}{\gt^2})(\max_{\gS}{\gx_-})^5} \right]^{1/4} .
	\end{equation}
Making such a choice for $\gb$ we have that $\gf_- = \gb \gx_-$ is a sub solution for \eqref{LichPlus}; i.e.,
	\begin{equation}\label{LichPlusSub}
	\UL{\gl} \gf_- \geq \gf_- - \tfrac18\gm^2 \gf_-^{-7} + \tfrac{1}{12}\gt^2 \gf_-^5. 
	\end{equation}

For later purposes, we note here that while the sub solution construction just described has been carried out for conformal data with the metric in the positive Yamabe class, in fact the same construction produces a sub solution for the other Yamabe classes as well. Note that for the construction to work in the other Yamabe classes, one still uses equation \eqref{SubLin1} to construct $\gz_-$, rather than an alternative form with $- \gz_-$ replaced on the left hand side by $+\gz_-$ or by zero.  One obtains, for any Riemannian metric $\gl$, a function 	$\gf_-$ satisfying \eqref{LichPlusSub}; that $\gf_-$ is a sub solution for the same equation with an appropriate change of sign (according to the Yamabe class of $\gl$) for the linear $\gf_-$ term immediately follows.

\Step{3}{Super solution}
A constant $\gf_+$ is a super solution if it satisfies the inequality
	\begin{equation}\label{ConstSup}
	\gf_+ - \tfrac18 \gm^2 \gf_+^{-7} + \tfrac{1}{12} \gt^2\gf_+^5 \geq 0. 
	\end{equation}
Clearly if one chooses $\gf_+$ to be $(\tfrac18 \max_{\gS}{\gm^2})^{1/8}$, then the inequality above is satisfied.  However, this choice does not guarantee that we have $\gf_+ \geq \gf_-$.  To ensure this latter condition, we choose
	\begin{equation}
	\gf_+ = \max{\{ 1, \tfrac18 \max_{\gS}{\gm^2}   \}} .
	\end{equation}
Recalling  (as determined in the previous step) that $\gf_- \leq \gx_- < \gz_-\leq A$, we verify that $\gf_- \leq \gf_+$.  One also readily verifies that this choice satisfies \eqref{ConstSup} for any $\gm$ and $\gt$.

\Step{4}{Existence of Solution}
Since $\gf_-$ and $\gf_+$ together constitute a set of smooth sub and super solutions which satisfy the hypotheses of Theorem \ref{T:SubSuper}, it follows that \eqref{LichPlus} has a smooth solution which is pointwise bounded by $\gf_{\pm}$.

\Step{5}{Uniqueness of Solution}
The sub and super solution theorem can be used to guarantee that a solution exists, but tells us nothing about the uniqueness of that solution.  To show that solutions of the Lichnerowicz equation are unique, we rely on the following lemma, proved in \cite{I-CMC}.
\begin{lemma}
Let $f: \gS \times\bbR \to\bbR$ be $C^1$ and satisfy
	\begin{equation}
	\frac{\del{}f}{\del{}s} (x,s) \gneqq 0
	\end{equation}
for all $x\in\gS$ and all $s\in I$, where $I$ is some interval (possibly infinite) in $\bbR_+$.  If $\gY_i$, $i=1,2$, are both solutions of
	\begin{equation}
	\UL{}\gY = f(x, \gY(x) )
	\end{equation}
and if $\gY_i$ take values in $I$ for all $x\in\gS$, then $\gY_1(x) = \gY_2(x)$ for all $x\in\gS$.
\end{lemma}\label{L:Unique}
For \eqref{LichPlus} we have
	\begin{equation}
	f(x,s) = s - \tfrac18 \gm^2 s^{-7} + \tfrac{1}{12} \gt^2 s^5, 
	\end{equation}
and therefore
	\begin{equation}
	 \begin{aligned}
	 \frac{\del{}f}{\del{}s} (x,s) &= 1 + \tfrac78 \gm^2 s^{-8} + \tfrac{5}{12} \gt^2 s^4\\
	 &> 0.
	\end{aligned}
	\end{equation}
Uniqueness of solutions to \eqref{LichPlus} follows immediately.
\end{proof}

A result for metrics in the zero Yamabe class follows from the work done to prove Proposition \ref{P:MuLichPlus}.  In particular, we obtain the following.

\begin{prop}\label{P:MuLichZero}
Let $\gS$ be a closed $3$-manifold; let $\gl$ be a smooth, zero Yamabe class metric on $\gS$; and let $\gm$ and $\gt$ be smooth functions on $\gS$ with $\gt$ nowhere zero and $\gm$ not identically zero.  Then there exists a unique smooth
 solution $\gf$ to the Lichnerowicz equation \eqref{MuLich}.
\end{prop}

\begin{proof}
It is sufficient to prove existence and uniqueness of solutions when $\Scal{\gl}=0$; thus we work with the equation
	\begin{equation}\label{LichZero}
	\UL{\gl} \gf = - \tfrac18 \gm^2 \gf^{-7} + \tfrac{1}{12} \gt^2 \gf^5. 
	\end{equation}

\Step{1}{Sub solution}
For the given choice of data $\{ \gS; \gl,\gm,\gt\}$ with zero Yamabe class metric $\gl$, we seek a function $\gy_-$ which satisfies the inequality
	\begin{equation}\label{LichZeroSub}
	 \UL{\gl} \gy_- \geq - \tfrac18 \gm^2 \gy_-^{-7} + \tfrac{1}{12} \gt^2 \gy_-^5. 
	\end{equation}
We have shown in the proof of Proposition \ref{P:MuLichPlus} that there is a function $\gf_->0$ which satisfies	
	\begin{equation}\label{RecallPlus}
	  \UL{\gl} \gf_- \geq \gf_- - \tfrac18 \gm^2 \gf_-^{-7} + \tfrac{1}{12} \gt^2 \gf_-^5. 
	\end{equation}
Since, as noted above, the argument for the existence of $\gf_-$ does not depend on the Yamabe class of the metric, and since such a function also satisfies \eqref{LichZeroSub},  we may take $\gy_- = \gf_-$.

\Step{2}{Super solution}
Any constant $\gy_+$ which satisfies the condition 
	\begin{equation}
	\gt^2 \gy_+^5 \geq \tfrac32 \gm^2 \gy_+^{-7} 
	\end{equation}
serves as a super solution for \eqref{LichZero}.  
A short computation shows that any constant $\gy_+$ satisfying
	\begin{equation}
	 (\min_{\gS}{\gt^2})^{1/12} \gy_+ \geq \max{\{1, \max_{\gS}{\gm^2} \}}
	\end{equation}
will be a super solution for \eqref{LichZero}.
To ensure as well that $\gy_+ \geq \gy_-$, we choose
	\begin{equation}
	\gy_+ = \left(\max{\{2, (\min_{\gS}{\gt^2})^{-1/12} \}} \right) \left( \max{\{1,\max_{\gS}{\gm^2}  \}}\right).
	\end{equation}
Note that we now require $\gt$ to be nowhere vanishing, unlike for metrics in the positive and negative Yamabe classes.

\Step{3}{Existence of Solutions}
Since $\gy_{\pm}$ constitute  smooth sub and super solutions for \eqref{LichZero}, if follows that \eqref{LichZero} has a smooth solution which is pointwise bounded by $\gy_-$ and $\gy_+$.

\Step{4}{Uniqueness}
We apply Lemma \ref{L:Unique} to the function
	\begin{equation}
	f(x,s) = -\tfrac18 \gm(x)^2 s^{-7} + \tfrac{1}{12} \gt(x)^2 s^5.
	\end{equation}
Since $\gt^2 >0$ we immediately see that $ \del{}f/\del{}s >0$; hence the solution to \eqref{LichZero} obtained is indeed unique.
\end{proof}

We use Propositions \ref{P:MuLichPlus} and \ref{P:MuLichZero} as key results for proving that the conformal method maps certain sets of near-CMC conformal data to solutions of the constraint equations, as we show in the next section (see Theorems 4.1 and 4.3).

\section{Near-CMC solutions of the constraint equations}
The semi-decoupled sequence method for obtaining near-CMC solutions of the coupled system \eqref{Lich}-\eqref{LW}, introduced in \cite{IM-nonCMC}, focusses on the sequence of equations
	\begin{align}\label{LichIterate}
	\UL{}  \Ugf{n}{} &=  \tfrac18 R\Ugf{n}{} - \tfrac18 ( \gs^{ab} + \LW{n}^{ab}) ( \gs_{ab} + \LW{n}_{ab}) \Ugf{n}{-7}
		+ \tfrac{1}{12} \gt^2 \Ugf{n}{5} 
		\\
	\DNa{}{a} (\LW{n})^a_b &= \tfrac23\, \Ugf{n-1}{6} \,\DNa{}{b} \gt. \label{LWIterate}
	\end{align}
Let us presume that we have made a specific choice of conformal data $\{ \gS; \gl, \gs, \gt \}$ satisfying appropriate hypotheses; for convenience we have in these equations suppressed explicit reference to the conformal metric $\gl$.
The idea is to iteratively define a sequence $\{(\Ugf{n}{}, \UW{n} )\}$ satisfying \eqref{LichIterate}-\eqref{LWIterate} by first choosing $\Ugf{0}{}$ arbitrarily\footnote{We do require that this choice of  $\Ugf{0}{}$ satisfy the inequality $\Dgf{\infty}{-} \leq   \Ugf{0}{} \leq  \Dgf{\infty}{+}$ as discussed below.}, then solving 
\eqref{LWIterate} with $n=1$ to obtain $\UW{1}$, then substituting $\UW{1}$ into  \eqref{LichIterate} with $n=1$ and solving  \eqref{LichIterate} for $ \Ugf{1}{}$, and thus proceeding to solve \eqref{LWIterate} and \eqref{LichIterate} alternately and iteratively so as to obtain the entire sequence $\{(\Ugf{n}{}, \UW{n} )\}$. Once the sequence is obtained, one proceeds to prove that it converges to a smooth limit $(\Ugf{\infty}{}, \UW{\infty} )$ which satisfies \eqref{Lich}-\eqref{LW}.  One finally  shows that, for a given choice of conformal data (and for any choice of $\Ugf{0}{}$), the solutions obtained are unique.

Our method for  showing that this sequence exists involves obtaining a sequence of sub solutions $\Dgf{n}{-}$ and super solutions $\Dgf{n}{+}$, as discussed in the last section. Both to show that these sub and super solutions exist and are controlled, and also to prove convergence of the sequence, we seek uniform upper and lower bounds for the set of all sub and super solutions, which consequently uniformly bound the sequence $\Ugf{n}{}$ itself; these  in turn imply uniform estimates for each $\UW{n}$. Once we find these uniform bounds $\Dgf{\infty}{-}$ and $\Dgf{\infty}{+}$, we have at our disposal the estimates
\begin{equation}
	 0< \Dgf{\infty}{-} \leq  \Dgf{n}{-} \leq  \Ugf{n}{} \leq \Dgf{n}{+}\leq \Dgf{\infty}{+} < \infty
	\end{equation}	
which hold for all $n=1,2,3, \dots$. In fact, in our construction, we may use $\Dgf{n}{+} = \Dgf{\infty}{+}$ ( i.e., $\Dgf{\infty}{+}$ is a super solution for \eqref{LichIterate} for all $n$), while we inductively show the existence of sub solutions $\{\Dgf{n}{-}\}$, which are uniformly bounded below by a positive constant $\Dgf{\infty}{-}$.

Unlike the Lichnerowicz equation \eqref{MuLich}, the coupled system \eqref{Lich}-\eqref{LW} is \emph{not} conformally covariant (see \cite{IM-nonCMC}).  Consequently, the clearest statement of our results here regarding the solvability of \eqref{Lich}-\eqref{LW} for a given set of conformal data involve two steps: We first state and prove solvability for conformal data with constant positive  curvature (Theorem \ref{PosConstThm}) and then use that result to prove a corollary for data with any positive Yamabe class metric (Corollary \ref{CorPos}). Similarly, we prove a solvability theorem for data including zero curvature metrics (Theorem \ref{ZeroConstThm}) and then extend the results to data with any zero Yamabe class metric.

\begin{theorem}
\label{PosConstThm}
Let $\gS$ be a closed three-dimensional manifold, let $\gl$ be a smooth Riemannian metric on $\gS$  which admits no conformal Killing fields
and has constant positive scalar curvature $\Scal{\gl} = +8$, and let $\gs$ be a smooth symmetric $2$-tensor on $\gS$ which is trace-free and divergence-free (with respect to $\gl$) and not identically zero. 
For every smooth function $\gt:\gS \to \bbR_+$ which satisfies the gradient conditions given by  \eqref{gradtau1} and 
\eqref{TauBound} and which also satisfies the gradient condition that the coefficient of $ \left| \Ugf{n}{} - \Ugf{n-1}{} \right|$ in equation \eqref{gradtau2} is sufficiently small, 
the equations \eqref{Lich}-\eqref{LW} with data $\{ \gS; \gl, \gs, \gt \}$ admit a unique smooth solution $(\gf, W)$. Consequently for every such set of data $\{ \gS; \gl, \gs, \gt \}$, there exists a unique solution $(\Sigma; \gamma, K)$ of the constraint equations \eqref{FirstEin}-\eqref{SecondEin}, taking the form  \eqref{gammarecon}-\eqref{Krecon}.
\end{theorem}

\begin{proof}
For conformal data of the sort hypothesized here, the semi-decoupled system \eqref{LichIterate}-\eqref{LWIterate} takes the form
	\begin{align}\label{LichPlusIt}
	\UL{}  \Ugf{n}{} &=  \Ugf{n}{} - \tfrac18 | \gs + \LW{n}|^2 \Ugf{n}{-7} + \tfrac{1}{12} \gt^2 \Ugf{n}{5} 
		\\
	\DNa{}{} \!\cdot\! L \UW{n}&= \tfrac23\, \Ugf{n-1}{6} \,\DNa{}{} \gt.\label{LWPlusIt}
	\end{align}

\Step{1}{Construction of the sequence}
We begin by choosing $\Ugf{0}{}$ such that $ \Dgf{0}{-}\leq\Ugf{0}{}\leq \Dgf{0}{+}$, for some constants $\Dgf{0}{\pm}$ to be chosen later (See the paragraph just before Step 2.) and which depend only on the choice of conformal data.  As is evident below, the value of $\Ugf{0}{}$ is irrelevant, provided it does satisfy the above inequality.  
The operator $\nabla\cdot L$ is elliptic and self-adjoint with respect to appropriate Sobelev spaces; under our assumption that $(\gS; \gl)$ admits no conformal Killing vector fields, it is also invertible.  Thus by standard elliptic theory (See, for example, the appendix of Besse \cite{B}.), the equation \eqref{LWPlusIt} with $n=1$ admits a unique solution $\UW{1}{}$, which as a consequence of  the smoothness of $\gt$ satisfies
	\begin{equation}
	\label{ellestim}
	\|  \UW{1}{} \|_{C^{k+2,\ga}} \leq c\, \| \Ugf{0}{6} \,\nabla \gt \|_{C^{k,\ga}},\qquad k \geq 0,\quad  \ga\in(0,1),
	\end{equation}
where $C^{l, \ga}$ denotes the $(l,\ga)$ H\"older norm of vector fields on $\gS$ given by $\gl$.  Furthermore, as argued in \cite{IM-nonCMC}, it follows from \eqref{ellestim} along with geometric considerations that there exists $C_S$, depending only on the Riemannian manifold $(\gS, \gl)$,  such that we have the pointwise estimate
	\begin{equation}
	|L\UW{1}| \leq C_S (\max_{\gS}{\Dgf{0}{+}})^6(\max_{\gS}{|\nabla\gt|});
	\end{equation}
similarly we find 
	\begin{equation}
	\label{ptwise}
	|L\UW{n+1}| \leq C_S (\max_{\gS}{\Dgf{n}{+}})^6(\max_{\gS}{|\nabla\gt|}).
	\end{equation}
	
We now describe how to choose uniformly bounded sub and super solutions $\Dgf{n}{-}, \Dgf{n}{+}$ for \eqref{LichIterate}.  This allows us to inductively construct a bounded sequence
 $\{(\Dgf{n}{}, \UW{n})\}$ satisfying \eqref{LichPlusIt}-\eqref{LWPlusIt}.  
 To this end we assume the existence of $\Dgf{n-1}{}$ such that $0 < \Dgf{n-1}{} \leq \Dgf{\infty}{+}$, for some constant $ \Dgf{\infty}{+}$.
By our inductive assumption, it follows that
	\begin{equation}\label{UniformLW}
	|L\UW{n}| \leq C_S (\max_{\gS}{\Dgf{\infty}{+}})^6(\max_{\gS}{|\nabla\gt|}).
	\end{equation}

We desire to choose $\Dgf{\infty}{+}$ so that it is a constant super solution for \eqref{LichPlusIt} for all $n\in \bbN $; it suffices that the estimate
	\begin{equation}
	0 \leq  \Dgf{\infty}{+} - \tfrac18 | \gs + \LW{n} |^2 \UDgf{\infty}{-7}{+} + \tfrac{1}{12} \gt^2 \UDgf{\infty}{5}{+}
	\end{equation}
holds on $\gS$.  For this to hold it suffices that
	\begin{equation}
	\UDgf{\infty}{8}{+} + \tfrac{1}{12} \gt^2 \UDgf{\infty}{12}{+} \geq \tfrac14 ( |\gs|^2 + | \LW{n}|^2), 
	\end{equation}
which in turn holds provided
	\begin{equation}
	 \UDgf{\infty}{8}{+} \geq \tfrac14 |\gs|^2 \quad \text{ and } \quad 
	 \tfrac{1}{12} \gt^2 \UDgf{\infty}{12}{+} 
	 \geq \tfrac{1}{4} C_S^2 \UDgf{\infty}{12}{+}
	 	\max_{\gS}{|\nabla\gt|^2}.
	\end{equation}
We now see that so long as we restrict $\gt$ so that
	\begin{equation}\label{gradtau1}
	 {3} C_S^2  	\left(\frac{\max_{\gS}{|\nabla\gt|}}{\min_{\gS}{\gt}} \right)^2 < 1,
	\end{equation}
it suffices to choose $\Dgf{\infty}{+}$ such that
	\begin{equation}
	\UDgf{\infty}{8}{+} \geq \tfrac14  \max_{\gS}|\gs|^2. 
	\end{equation} 
	
We turn to the task of finding a sequence of sub solutions $\Dgf{n}{-}$, which are defined to be a sequence of functions such that
	\begin{equation}\label{YPlusSub}
	\UL{} \Dgf{n}{-} \geq \Dgf{n}{-} - \tfrac18 | \gs + \LW{n}|^2 \UDgf{n}{-7}{-} 
		+ \tfrac{1}{12} \gt^2 \UDgf{n}{5}{-},
	\end{equation}
where $\UW{n}$ satisfies \eqref{LWPlusIt}.  Here, we seek a positive constant $\Dgf{\infty}{-}$, independent of $n$, such that  $\Dgf{\infty}{+} \geq \Dgf{n}{-}\geq \Dgf{\infty}{-}>0$.  

Following the method used in the proof of Proposition \ref{P:MuLichPlus}, we first study the solution $\Ugy{n}{}$ to
	\begin{equation}
	\gD \Ugy{n}{} - \Ugy{n}{} = -\tfrac18 | \gs + \LW{n}|^2 \UA{n}{-7} ,
	\end{equation}
where $\UA{n}{} = \max{\{1, \tfrac18 \max_\Sigma { | \gs + \LW{n}|^2}\}}$.  We estimate $\Ugy{n}{}$ using $\gx$, the (smooth) solution to
	\begin{equation}
	\gD \gx - \gx = -\tfrac18|\gs|^2 A^{-7}, 
	\end{equation}
where $A = \max{\{ 1, \max_\Sigma{|\gs|^2}\}}$.  By the maximum principle\footnote{See, for example, version 2 in \cite{I-CMC}.} and the compactness of $\gS$, there exists a constant $\gd>0$, depending on $( \gS, \gl, \gs)$, such that $\gx \geq \gd$.

We claim that one may choose, depending only on $\Dgf{\infty}{+}$ (and hence on $\max_{\Sigma}{|\gs|^2}$), a constant $C_{\gt}$ such that the condition
	\begin{equation}\label{TauBound}
	|\nabla\gt| \leq C_{\gt} 
	\end{equation}
implies that $| \Ugy{n}{} - \gx |$ is small enough to ensure that $\Ugy{n}{} \geq \tfrac12 \gd$.  The claim follows from examining the equation
	\begin{equation}\label{SubDiff}
	 (\gD -1)(\Ugy{n}{} - \gx) =  -\tfrac18 | \gs + \LW{n}|^2 \UA{n}{-7} + \tfrac18|\gs|^2A^{-7}.
	\end{equation}
	
We first assume that $A=1$ and $\UA{n}{}=1$.  In this case, we may write the right side of \eqref{SubDiff} as $F(\LW{n}) - F(0)$, where $F(\gr) = -\tfrac18 |\gs + \gr|^2$.  Applying the mean value theorem, and making use of  \eqref{UniformLW}, we see that the right side of \eqref{SubDiff} is controlled by
	\begin{equation}
	\tfrac14 |\nabla\gt| \left( C_S^2 \,\UDgf{\infty}{12}{+}\max_{\gS}|\nabla\gt| +C_S\, \UDgf{\infty}{6}{+}\max_{\gS}{|\gs|}\right)\left(\max_{\gS}{|\gs|} \right).
	\end{equation}	
Thus by the maximum principle\footnote{Version 3 in \cite{IM-nonCMC}.}, we have $|\Ugy{n}{} - \gx |$ is small whenever $|\nabla\gt|$ is small.

In the case that $A >1$, we can choose $C_{\gt}$ small so that \eqref{TauBound} implies $\UA{n}{} >1$ for all $n$.  Then the right side of \eqref{SubDiff} is equal to
	\begin{equation}
	 -\frac18 \frac{ |\gs + \LW{n}|^2 - |\gs|^2}{\left(\max_{\gS}{|\gs + \LW{n}|^2}\right)^{7}}
	 -\frac18|\gs|^2 \left[\left(\max_{\gS}{|\gs + \LW{n}|^2}\right)^{-7} - \left(\max_{\gS}{|\gs|^2}\right)^{-7} \right].
	\end{equation}
Making use of \eqref{UniformLW} once again, we see that this quantity can be made small by controlling $|\nabla\gt|$.  Thus an application of the maximum principle yields the claim in this case.

With the claim in hand, one easily verifies that $-\gD\Ugy{n}{} + \Ugy{n}{} \leq \UA{n}{}$ and hence by the maximum principle we have $ \Ugy{n}{} \leq \UA{n}{}$.  From this it follows that
	\begin{equation}
	\gD \Ugy{n}{} \geq \Ugy{n}{} - \tfrac18 | \gs + \LW{n} |^2 \Ugy{n}{-7}. 
	\end{equation}
Replacing $\Ugy{n}{}$ by $\Ugy{n}{} - \tfrac14 \gd$, we see that the previous estimate holds with a strict inequality.

We now choose  a constant $\gb \in (0,1)$, independent of $n$, so that $\Dgf{n}{-} := \gb \Ugy{n}{}$ is a sub solution for \eqref{LichIterate}.  One verifies, using an argument similar to that in the proof of Proposition \ref{P:MuLichPlus}, that for any choice of $\gb\in(0,1)$ satisfying
	\begin{equation}
	\gb^5 \leq \frac{3\gd}{(\max_{\gS}{\gt^2}) \,\UDgf{\infty}{5}{+} },
	\end{equation}
we have
	\begin{equation}
	\gD(\gb \Ugy{n}{}) \geq (\gb\Ugy{n}{}) - \tfrac18 | \gs + \LW{n} |^2 (\gb\Ugy{n}{})^{-7} + \tfrac{1}{12} \gt^2 (\gb\Ugy{n}{})^{5}. 
	\end{equation}
Note that any such sub solution $\beta \Ugy{n}{}$ is  bounded below by $\Dgf{\infty}{-}:=\tfrac15 \gb \gd>0$, independently of $n$.  

Finally, if necessary, we choose a larger $\Dgf{\infty}{+}$ to ensure $\Dgf{\infty}{-} < \Dgf{\infty}{+}$.  This allows construction of the sequence $\{(\Ugf{n}{}, \UW{n} )\}$ such  that 
	\begin{equation}
	  0 < \Dgf{\infty}{-} \leq\Ugf{n}{}\leq \Dgf{\infty}{+}
	\end{equation}
and
	 \begin{equation}
	| \LW{n} | \leq C_S \,\UDgf{\infty}{6}{+} \max_{\gS}{|\nabla\gt|}.
	\end{equation}
We subsequently set $\Dgf{0}{-} = \Dgf{\infty}{-}$ and $\Dgf{0}{+} = \Dgf{\infty}{+}$.

\Step{2}{Convergence of the sequence}
We now show that the sequence $\{ (\Ugf{n}{}, \UW{n} )\}$ converges to a smooth limit $(\Ugf{\infty}{}, \UW{\infty})$.  Standard elliptic estimates applied to
	\begin{equation}
	(\nabla\! \cdot \! L)(\UW{n} - \UW{m}) = \tfrac23 [ \Ugf{n}{6} - \Ugf{m}{6} ] \nabla \gt
	\end{equation}
imply that $\{ \UW{n}\}$ is Cauchy in $W^{2,p}(\gS)$, with $p>3$, provided $\{\Ugf{n}{}\}$ is Cauchy in $C^0(\gS)$.  In light of the Sobolev embedding $W^{2,p}(\gS)\subset C^0(\gS)$ (See, for example, the appendix of Besse \cite{B}.), the sequence $\{ (\Ugf{n}{}, \UW{n} )\}$ converges to $(\Ugf{\infty}{}, \UW{\infty}) \in C^0(\gS) \times C^0(\gS)$ provided $\{\Ugf{n}{}\}$ converges.  \mnote{Return to this once the analysis section is written}
Thus we turn our attention to this sequence.

We study this sequence by considering the  quantity
	\begin{equation}
	\label{scrI}
	\scI(x,\Ugf{n-1}{}, \Ugf{n}{}, \Ugf{n+1}{}) := \int_0^1 \frac{d}{dt} \left[ \gD \Ugy{n+1}{}(t,x) - F(x,\Ugy{n}{}(t,x), \Ugy{n+1}{}(t,x)) \right] \,dt,
	\end{equation}
where $\Ugy{n}{}(t,x) = t \Ugf{n}{}(x) + (1-t)\Ugf{n-1}{}(x)$ and where 
	\begin{equation}
	\label{F}
	F (x,\Ugy{n}{}(t), \Ugy{n+1}{}(t)): = \tfrac18 \Ugy{n+1}{} - \tfrac18 |\gs +  \LV{n}|^2 \Ugy{n+1}{-7} + \tfrac{1}{12}\gt^2 \Ugy{n+1}{5}. 
	\end{equation}
Here the vector field $\UV{n}$ satisfies 
	\begin{equation}
	\nabla\!\cdot\! L \UV{n} = \tfrac23 \Ugy{n}{6}\nabla\gt
	\end{equation}
and we have suppressed dependence on the point $x\in \gS$.  Computing the quantity $\scI$ via the Fundamental Theorem of Calculus and also by direct computation, we obtain
	\begin{equation}
	\gD(\Ugf{n+1}{} - \Ugf{n}{}) - \scG [\Ugf{n+1}{} - \Ugf{n}{}]
		= \scF[\Ugf{n}{} - \Ugf{n-1}{}],
	\end{equation}
where
	\begin{equation}
	\label{FG}
	\begin{aligned}
	\scF[\Ugf{n}{} - \Ugf{n-1}{}] &= \int_0^1D_2 F(\cdot,\Ugy{n}{}(t), \Ugy{n+1}{}(t))\,dt\,[\Ugf{n}{} - \Ugf{n-1}{}] \\
	\scG[\Ugf{n+1}{} - \Ugf{n}{}] &=  \int_0^1D_3 F(\cdot,\Ugy{n}{}(t), \Ugy{n+1}{}(t))\,dt\,[\Ugf{n+1}{} - \Ugf{n}{}];
	\end{aligned} 
	\end{equation}
here $D_i$ is differentiation with respect to the $i^{\text{th}}$ variable.  One easily sees that $\scG$ satisfies
	\begin{equation}
	 \scG [\Ugf{n+1}{} - \Ugf{n}{}] \geq \tfrac18 (\Ugf{n+1}{} - \Ugf{n}{}).
	\end{equation}
An estimate for $\scF$ can be obtained by observing that
	\begin{equation}
	 \scF[\Ugf{n}{} - \Ugf{n-1}{}] =\tfrac14 \int_0^1 \left( \gs^{ab} + LV[\Ugy{n}{}(t)]^{ab}\right)
	 	\left(L\go(t)[\Ugf{n}{} - \Ugf{n-1}{}]_{ab}\right) \Ugy{n+1}{-7}(t) \,dt,
	\end{equation}
where the vector field $\go[\Ugf{n}{} - \Ugf{n-1}{}]$ is defined to be the solution to
	\begin{equation}
	\nabla\!\cdot\! L \go[\Ugf{n}{} - \Ugf{n-1}{}] =  4 \Ugy{n}{5} \nabla\gt (\Ugf{n}{} - \Ugf{n-1}{}).
	\end{equation}
Thus by standard elliptic estimates used above, we see that
	\begin{equation}
	\label{gradtau2}
	\left|  \scF[\Ugf{n}{} - \Ugf{n-1}{}]  \right| \leq
		\tfrac14 \left(\max_{\gS}{|\gs|} + \hat{C} \tfrac43 \UDgf{\infty}{6}{+}\max_{\gS}{|\nabla\gt|} \right)  
		\left(8 \hat{C} \UDgf{\infty}{5}{+}\max_{\gS}{|\nabla\gt|} \right)
		\left(\UDgf{\infty}{-7}{+}\right) 
		\cdot  \left| \Ugf{n}{} - \Ugf{n-1}{} \right| .
	\end{equation}
Hence, provided $\max_{\gS}{|\nabla\gt|}$ is small, we obtain via the maximum principle that
	\begin{equation}
	\label{contraction}
	 \left| \Ugf{n+1}{} - \Ugf{n}{} \right|\leq \gL \left| \Ugf{n}{} - \Ugf{n-1}{} \right|
	\end{equation}
for some positive constant $ \gL <1$.  It follows that  $ \{ \Ugf{n}{} \}$ is a Cauchy sequence in $C^0(\gS)$, and converges to $\Ugf{\infty}{}\in C^0(\gS)$. Further, as a consequence of  the argument described above, there is a $C^0$ limit $\UW{\infty}$ for the sequence $\{ \UW{n}\}$.  We may then adapt an argument from  \cite{IM-nonCMC}, together with the smoothness of the data, to verify  that in fact  $\Ugf{\infty}{}$ and $\UW{\infty}$ are smooth.

\Step{3}{Showing that the limit is the unique solution}
To see that $(\Ugf{\infty}{}, \UW{\infty})$ constitutes a solution to \eqref{Lich}-\eqref{LW}, it suffices to observe that for $p>3$ as above, we have the estimate
	\begin{equation}
	\left\| \gD \Ugf{n}{} - F(\cdot, \Ugf{n}{}, \Ugf{n}{}) \right\|_{W^{0,p}} \leq C_{\text{Sobolev}} 
		\left\| F(\cdot, \Ugf{n-1}{}, \Ugf{n}{}) - F(\cdot, \Ugf{n}{}, \Ugf{n}{}) \right\|_{C^{0}}.
	\end{equation}
The continuity of $F$ implies that the right side tends to zero as $n \to \infty$; thus we see that $\Ugf{\infty}{}$ is a weak solution to \eqref{Lich}.  The smoothness of $\Ugf{\infty}{}$ implies that this weak solution is in fact a classical solution; from a similar argument we obtain that \eqref{LW} is  weakly (and therefore classically) satisfied.

To verify uniqueness, let $(\gf, W)$ and $(\widehat{\gf}, \widehat{W})$ be a pair solutions.  Define $\gF(t) = t \gf + (1-t)\widehat{\gf}$ and consider the quantity
	\begin{equation}
	\scK (x, \gf, \widehat{\gf} )= \int_0^1 \frac{d}{dt} \left[ \gD \gF - F(x, \gF, \gF ) \right] \,dt.
	\end{equation}
The analysis used to show convergence of $(\Ugf{n}{}, \UW{n})$ yields the estimate
	\begin{equation}
	 \| \gf - \widehat{\gf} \|_{C^0} \leq \gL 	\| \gf - \widehat{\gf} \|_{C^0},
	\end{equation}
where, as before, $\gL <1$.  Thus $\gf = \widehat{\gf}$, from which it follows immediately that $W = \widehat{W}$.
\end{proof}

It is possible to use this theorem, in conjunction with the Yamabe theorem, to find a solution to the constraint equations in the conformal class of any Yamabe positive metric $\gl$ on $\gS$.  Given a Yamabe positive Riemannian manifold $(\gS, \gl)$, by the Yamabe theorem we can find a smooth, positive function $\gq$ such that $(\gS, \gq^4\gl)$ has scalar curvature $R =8$.  Then by applying  Theorem \ref{PosConstThm} to the system
	\begin{align}\label{YamabeLich}
	\UL{\gq^4\gl}\gf &= \tfrac18 \Scal{\gq^4\gl}\gf - \tfrac18 ( \gs + \LW{\gq^4\gl})^2 \gf^{-7} + \tfrac{1}{12} \gt^2 \gf^5 \\
	\Un{\gq^4\gl}\cdot \LW{\gq^4\gl} &= \gf^6 \, \Un{\gq^4\gl}\gt, \label{YamabeLW}
	\end{align}
we arrive at the following corollary.
\begin{corollary}
\label{CorPos}  
Suppose that $\gS$ is a three-dimensional Riemannian manifold with smooth metric $\gl$ in the Yamabe positive class and having no conformal Killing vector fields.  For each smooth symmetric $2$-tensor $\gs$  which is trace-free and divergence-free with respect to $\gl$, and for each smooth function $\gt: \gS \to \bbR$ which is non zero and which satisfies the hypotheses of 
Theorem \ref{PosConstThm}, there  exist  smooth positive functions $\gf$ and $\gq$ and a smooth vector field $W$ such that the data
	\begin{equation}
	\begin{aligned}
	\gg_{ab} &= (\gf\gq)^4 \gl_{ab} \\
	K^{ab} &= \gf^{-10} \left( \gq^{-10}\gs + \LW{\gq^4\gl} \right)^{ab} + \tfrac13 (\gf\gq)^{-4}\gl^{ab}\gt,
	\end{aligned} 
	\end{equation}
comprise a solution to the Einstein constraint equations.
\end{corollary}

We now turn our attention to the case of conformal data which includes metrics which lie in the zero Yamabe class (i.e., metrics which can be conformally transformed to a metric with zero scalar curvature).  We start by proving a result for metrics with $R=0$:

\begin{theorem}
\label{ZeroConstThm}
Let $\gS$ be a closed three-dimensional manifold, let $\gl$ be a smooth Riemannian metric on $\gS$  which admits no conformal Killing fields
and which has identically vanishing scalar curvature, and let $\gs$ be a smooth symmetric $2$-tensor on $\gS$ which is trace-free and divergence-free (with respect to $\gl$) and not identically zero. 
For every smooth function $\gt:\gS \to \bbR_+$ which is nowhere zero and which satisfies the gradient conditions \eqref{Z:GradTauOne} and \eqref{TauBound} and which also satisfies the gradient condition that the coefficient of 
$ \left| \Ugf{n}{} - \Ugf{n-1}{} \right|$ in equation \eqref{gradtau2} is sufficiently small,
the equations \eqref{Lich}-\eqref{LW} with data $\{ \gS; \gl, \gs, \gt \}$ admit a unique smooth solution $(\gf, W)$. Consequently for every such set of data $\{ \gS; \gl, \gs, \gt \}$, there exists a unique solution $(\Sigma; \gamma, K)$ of the constraint equations \eqref{FirstEin}-\eqref{SecondEin}, taking the form  \eqref{gammarecon}-\eqref{Krecon}.
\end{theorem}

\begin{proof}
The proof of Theorem \ref{ZeroConstThm} (Yamabe zero metrics) is very much like that of  Theorem \ref{PosConstThm} (positive Yamabe class metrics). We sketch the steps here, emphasizing the differences from the proof of Theorem \ref{PosConstThm}. 

\Step{1}{Construction of the sequence}
We make use of Proposition \ref{P:MuLichZero} to construct a sequence $\{(\Ugf{n}{}, \UW{n})\}$ satisfying
	\begin{align}\label{LichZeroIt}
	\UL{}  \Ugf{n}{} &= - \tfrac18 | \gs + \LW{n}|^2 \Ugf{n}{-7} + \tfrac{1}{12} \gt^2 \Ugf{n}{5} 
		\\
	\DNa{}{} \!\cdot\! L \UW{n}&= \tfrac23\, \Ugf{n-1}{6} \,\DNa{}{} \gt \label{LWZeroIt},
	\end{align}
and in doing so we obtain uniform upper and lower bounds for the corresponding sequence of sub and super solutions.  As in the positive curvature case, we may choose a constant $\Dgf{\infty}{+}$ which is a super solution for \eqref{LichZeroIt} for all $n$.  It suffices that the constant $\Dgf{\infty}{+}$ satisfy
	\begin{equation}\label{Z:SuperBd}
	\frac{\max_{\gS}{|\gs|^2}}{3\min_{\gS}{\gt^2}} \leq \UDgf{\infty}{12}{+},  
	\end{equation}
provided the estimate
	\begin{equation}
	|\LW{n}|^2 \leq \tfrac13 \gt^2 \UDgf{\infty}{12}{+}
	\end{equation}
also holds.  In light of the elliptic estimate
	\begin{equation}
	 |\LW{n}| \leq C_S\,\max_{\gS}{\Ugf{n-1}{6}}\,\max_{\gS}{|\nabla\gt|},
	\end{equation}
a constant $\Dgf{\infty}{+}$ satisfying \eqref{Z:SuperBd} is a super solution so long as we require that the conformal data satisfy the condition
	\begin{equation}\label{Z:GradTauOne}
	\frac{\max_{\gS}{|\nabla\gt|^2}}{\min_{\gS}{\gt^2}} \leq \frac{1}{3C_S^2}.
	\end{equation}
	
We now show that the sequence of sub solutions $\{\Dgf{n}{-}\}$ provided by the proof of Proposition \ref{P:MuLichZero} is bounded below by a positive function $\Dgf{\infty}{-}$.  Recall from the proof that the sub solutions for equation \eqref{LichZero} were in fact chosen to satisfy \eqref{RecallPlus}.  Thus the sub solutions for \eqref{LichZeroIt} in fact satisfy \eqref{YPlusSub} and hence, by the proof of Theorem \ref{PosConstThm}, are indeed uniformly bounded below by a positive constant $\Dgf{\infty}{-}$.  With $\Dgf{\infty}{-}$ in hand, we may increase $\Dgf{\infty}{+}$ if necessary to ensure that, for all $n$, the following holds
	\begin{equation}
	0 < \Dgf{\infty}{-} \leq \Dgf{n}{-} \leq \Ugf{n}{} \leq \Dgf{n}{+} \leq \Dgf{\infty}{+} < \infty.
	\end{equation}
Recall that the argument constructing $\Dgf{\infty}{-}$ places conditions on the size of $\nabla\gt$.

\Step{2}{Convergence of the sequence}
The argument for convergence of the sequence here (with $R=0$ conformal data) is very similar to that given for convergence of the sequence in the proof of Theorem \ref{PosConstThm} ($R>0$ conformal data). We define $\scJ$  
much as in \eqref{scrI} (with the $\frac{1}{8} \Ugy{n+1}{}$ term subtracted from the quantity $F$ in \eqref{F}). Then to obtain a contraction map of the form \eqref{contraction} (with $\Lambda<1$ ), we need to carry out the estimates for the operators $\scF$ and $\scG$ as in \eqref{FG}. 

The estimate for $\scF$ is precisely the same as above; we obtain \eqref{gradtau2}. For $\scG$ we easily calculate that 
	\begin{equation}
	 \scG [\Ugf{n+1}{} - \Ugf{n}{}] \geq \tfrac{5}{12}\min_{\gS}{\gt^2}\,\Dgf{\infty}{-}\, (\Ugf{n+1}{} - \Ugf{n}{}).
	\end{equation}
Combining these estimates, we readily determine that for sufficiently small max $|\nabla \tau|$, we have  \eqref{contraction} with $\Lambda<1$.  The convergence of the sequence $\{ (\Ugf{n}{}, \UW{n} )\}$ then follows,  as in the proof of Theorem  \ref{PosConstThm}.

\Step{3}{Showing that the limit is the unique solution}
The argument that the limit of the sequence $\{ (\Ugf{n}{}, \UW{n} )\}$ is a smooth solution, and that it is unique for the given set of conformal data, proceeds exactly as in the proof of Theorem  \ref{PosConstThm}.

\end{proof}

Combining this result with the Yamabe theorem for metrics of the zero Yamabe class, we produce (analogous to Corollary \ref{CorPos}) the following:

\begin{corollary} 
\label{CorZero}
Suppose $\gS$ is a three-dimensional Riemannian manifold with smooth metric $\gl$ in the Yamabe zero class having no conformal Killing vector fields.  For each smooth symmetric $2$-tensor $\gs$  which is trace-free and divergence-free with respect to $\gl$, and for each smooth function $\gt: \gS \to \bbR$ which is nowhere  zero and which satisfies the hypotheses of 
Theorem \ref{ZeroConstThm}, there  exist  smooth positive functions $\gf$ and $\gq$ and a smooth vector field $W$ such that the data
	\begin{equation}
	\begin{aligned}
	\gg_{ab} &= (\gf\gq)^4 \gl_{ab} \\
	K^{ab} &= \gf^{-10} \left( \gq^{-10}\gs + \LW{\gq^4\gl} \right)^{ab} + \tfrac13 (\gf\gq)^{-4}\gl^{ab}\gt,
	\end{aligned} 
	\end{equation}
is a solution to the Einstein constraint equations.
\end{corollary}

\section{Conclusions}

The results we present here, together with those of the earlier papers  \cite{IM-nonCMC}, and \cite{IO},  provide a fairly complete picture of which sets of near-CMC conformal data on compact manifolds lead to solutions of the Einstein constraint equations and which do not. Similarly, the picture for near-CMC asymptotically Euclidean data \cite{CBIY-AE}  and for near-CMC asymptotically hyperbolic data \cite{IP-AH} is fairly clear as well. Even for the case of near-CMC data on manifolds with boundary, the recent results of  \cite{HKN} point toward increasing clarity.

On the other hand, almost nothing is understood about conformal data which is neither CMC nor near-CMC. This is the direction which future research into the use of the conformal method for obtaining solutions of the Einstein constraint equations is bound to explore.\footnote{We note very recent work by Holst, Nagy, and Tsogtgerel which appears to describe certain sets of conformal data which are not near-CMC and do admit solutions of the coupled Einstein constraint equations. To date, these data sets are all positive Yamabe, have small $|\gs|$, and necessarily have non vanishing (but small) matter density present.} 

\section{Acknowledgments}
We thank both the Albert Einstein Institute (Golm) and IHES (Bures) for providing a conducive research environment for the writing of portions of this work, and we thank the referee for very helpful comments.  Partial support for this work has been provided by NSF grants PHY-0354659 and PHY-0652903 to the University of Oregon.


\end{document}